\pdfoutput=1
\documentclass[10pt,acmsmall,screen,nonacm]{acmart}
\usepackage{natbib}

\makeatletter
\let\@authorsaddresses\@empty
\makeatother

\usepackage[algo2e,ruled,lined,linesnumbered,noend]{algorithm2e}
\usepackage[noend]{algpseudocode}

\usepackage[compact]{titlesec}
\titlespacing*{\section}
{0pt}{.9ex plus .5ex minus .2ex}{.3ex plus .1ex}
\titlespacing*{\subsection}
{0pt}{.7ex plus .5ex minus .2ex}{.3ex plus .1ex}

\titleformat{\section}{\Large\bfseries}{\thesection}{1em}{}
\titleformat{\subsection}{\large\bfseries}{\thesubsection}{1em}{}
\titleformat{\subsubsection}[runin]{\bfseries}{\thesubsubsection}{1em}{}

\usepackage{pgfplots}
\pgfplotsset{compat=1.17}
\usepackage{mathtools}
\usepackage[shortlabels]{enumitem}
\usepackage{color}
\usepackage{tikz}

\usepackage{hyperref}
\hypersetup{
    colorlinks=true,
    linktoc=all,
    allcolors=blue,
}

\SetArgSty{textrm}
\SetKwProg{Fn}{Function}{}{}

\DeclareMathOperator{\poly}{polylog}
\DeclareMathOperator{\supp}{supp}

\newcommand{\E}{\mathbb{E}}
\newcommand{\mf}[1]{{\mbox{\sc{#1}}}}

\newcommand{\inv}{\mathrm{inv}}

\algdef{SE}[SUBALG]{Indent}{EndIndent}{}{\algorithmicend\ }
\algtext*{Indent}
\algtext*{EndIndent}

\numberwithin{equation}{section}

\theoremstyle{plain}
\newtheorem{thm}[equation]{Theorem}
\newtheorem{cor}[equation]{Corollary}
\newtheorem{lem}[equation]{Lemma}

\theoremstyle{definition}
\newtheorem{defn}[equation]{Definition}

\theoremstyle{remark}

\newcommand{\whp}{\emph{whp}}

\begin{document}

\title{Encoding Schemes for Parallel In-Place Algorithms}

\author{Chase Hutton}
\affiliation{
    \institution{University of Maryland, College Park}
    \country{United States}
}
\author{Adam Melrod}
\affiliation{
    \institution{University of Maryland, College Park}
    \country{United States}
}


\begin{abstract}
    Many parallel algorithms which solve basic problems in computer science use auxiliary space linear in the input to facilitate conflict-free computation. There has been significant work on improving these parallel algorithms to be in-place, that is to use as little auxiliary memory as possible. In this paper, we provide novel in-place algorithms to solve the fundamental problems of merging two sorted sequences, and randomly shuffling a sequence. Both algorithms are work-efficient and have polylogarithmic span. Our algorithms employ encoding techniques which exploit the underlying structure of the input to gain access to more bits, which enables the use of auxiliary data as well as non-in-place methods. The encoding techniques we develop are general. We expect them to be useful in developing in-place algorithms for other problems beyond those already mentioned. To demonstrate this, we outline an additional application to integer sorting. In addition to our theoretical contributions, we implement our merging algorithm, and measure its memory usage and runtime.
\end{abstract}

\maketitle

\section{Introduction} There is a growing interest in devising theoretically efficient parallel algorithms with low memory footprints. Such algorithms enjoy a number of benefits over their non-in-place counterparts including reduced
memory allocation cost, increased cache efficiency, and greater scalability. Recently, theoretically-efficient and practical parallel in-place (PIP) algorithms have been proposed for a wide variety of individual problems including radix sort \cite{10.1145/3323165.3323198}, partition \cite{kuszmaul2020inplaceparallelpartitionalgorithmsusing}, and random permutation \cite{penschuck2023engineeringsharedmemoryparallelshuffling}. In many cases, these algorithms achieve better performance than the best existing alternatives.

Gu et al. generalize the ideas of previous work on PIP into a new framework; they give two new models, which they refer to as the strong PIP model and the relaxed PIP model \cite{2021pip}. In both models, the auxiliary space is measured in the sequential setting. This decouples the analysis of a parallel algorithm's span from its space complexity and reflects the memory required to initialize parallelism in many modern parallel computing environments (e.g. the `busy leaves' property \cite{doi:10.1137/S0097539793259471}). Formally, the strong and relaxed PIP models assume fork-join parallel computations using $O(\log n)$-word auxiliary space in a stack-allocated fashion for inputs of size $O(n)$ when run sequentially. The relaxed PIP model additionally allows the use of $O(n^{1-\epsilon})$ (for fixed $\epsilon > 0$) heap-allocated memory. We overview the strong PIP model in Section \ref{strong PIP}. For a full discussion of the relaxed PIP model, refer to Section 3.2 of \cite{2021pip}.

The strong PIP model is restrictive due to the polylogarithmic auxiliary space requirement. As a consequence, there are only a small number of known non-trivial work-efficient strong PIP algorithms: unstable partition \cite{kuszmaul2020inplaceparallelpartitionalgorithmsusing}, scan, reduce, rotation, and set operations \cite{2021pip}, and constructing implicit search tree layouts \cite{9416144}. The relaxed PIP model, on the other hand, captures the trade off between additional space and span exhibited by many existing PIP algorithms. In particular, the allowance of $O(n^{1-\epsilon})$ auxiliary space enables the possibility of iteratively reducing a problem to a subproblem of sufficiently small size. Gu et al. observe that many problems of interest satisfy this `decomposable property.' In particular, they give a number of work-efficient relaxed PIP algorithms including algorithms for random permutation, merging, filter, and list contraction. They also implement these algorithms. In all cases, their implementations exhibited competitive or even better performance compared to their non-in-place counterparts, along with using less space \cite{2021pip}. Ultimately, the relaxed PIP model is an important step in the direction of obtaining efficient low space parallel algorithms. Still, while performant, the $O(n^{1-\epsilon})$ auxiliary space and $O(n^{\epsilon}\poly n)$ span used by relaxed PIP algorithms can be undesirable.

The primary goal of this paper is to enlarge the set of work-efficient strong PIP algorithms by developing a set of encoding techniques, and applying them to solve a number of canonical problems in algorithm design. The `encoding techniques' we use take advantage of underlying structure within the input in order to `store' additional information. We primarily use two encoding schemes:

\subsubsection{Inversion encoding.} A central idea in this paper is that of \textit{inversion encoding}. This technique is also known as \textit{bit-stealing} \cite{IANMUNRO198666}. The high-level idea is that when our input is taken from a totally ordered universe (e.g. integers), we can divide the input into pairs and exploit the ordering to gain extra bits of information. These techniques functionally allow the use of a sub-linear (but still significant) amount of auxiliary memory, roughly $\Omega(n)$ bits. The main barrier to using inversion encoding is the overhead of reading and writing `stolen' bits. For example, reading and writing an encoded $O(\log n)$-bit value requires $O(\log n)$ work. We provide additional issues that arise when using inversion encoding in Section \ref{encodingdisc}. Despite the limitations of the technique, we show that inversion encoding is a powerful tool in designing in-place parallel algorithms.

\subsubsection{Restorable buffers.} As reading and writing encoded values requires $O(\log n)$ work, an alternative that does not is desirable. One such alternative is a construction we call \emph{restorable buffers}. These buffers use inversion encoding to encode a range of values from the input into a different predetermined part of the input, freeing the initial range to be used as an auxiliary workspace with constant reads and writes. Additionally, when the auxiliary workspace has fulfilled its purpose, the values that originally occupied that portion of the array can be restored. We give a detailed description of this construction in Section \ref{restorable-buffers}.

\subsubsection{Our Contributions.} The main contributions of this paper are work-efficient, polylogarithmic span strong PIP algorithms for merging two sorted sequences and randomly shuffling a sequence. These results are new to the best of our knowledge. We discuss these algorithms in Sections \ref{merge section} and \ref{randperm}. Ultimately, we expect the techniques we have developed to be applicable in the development of parallel in-place algorithms beyond those we discuss in this paper; a further application to integer sorting is discussed with less detail in Section \ref{add-apps}.

We implement our in-place parallel merging algorithm and compare it against the optimized non-in-place merging implementation in the Problem Based Benchmark Suite (PBBS) \cite{10.1145/2312005.2312018}. The
running time comparisons for certain input sizes are shown in Figure \ref{run time comparison} and we provide more details in Section \ref{exp section}. The data illustrates that in addition to having negligible memory usage, our in-place merging algorithm can be competitive with a state of the art implementation. 






\begin{figure}[ht]
\scalebox{0.8}{%
  \begin{minipage}{\textwidth}
\begin{tikzpicture}
\begin{axis}[
    width=12cm,
    height=7cm,
    ylabel={Time (\(\mu\)s)},
    ymajorgrids=true,
    major x tick style = transparent,
    enlarge x limits=0.2,
    symbolic x coords={1M,25M,50M,100M,250M,500M,1B,2B},
    xlabel={Input Size},
    xtick=data,
    scaled y ticks=false,
    legend cell align={left},
    legend style={
        at={(0.5,1.03)},
        anchor=south,
        legend columns=2,
        /tikz/every even column/.append style={column sep=0.5cm}
    },
]

\addplot+[mark=*, very thick, color=gray] coordinates {
    (1M,  115)
    (25M, 2240)
    (50M, 4217)
    (100M, 6233)
    (250M, 16882)
    (500M, 59298)
    (1B,   63496)
    (2B,  173050)
};

\addplot+[mark=triangle, very thick, color=red] coordinates {
    (1M,   1683)
    (25M,  8179)
    (50M, 11941)
    (100M, 19741)
    (250M, 57307)
    (500M, 92210)
    (1B,   193711)
    (2B,   341541)
};

\legend{PBBS Merging, IP Merging}

\end{axis}
\end{tikzpicture}

\caption{Running times for our parallel in-place merging (IP Merging) implementation compared to non-in-place implementation from
PBBS \cite{10.1145/2312005.2312018}. Here the IP merging implementation uses blocks of size 4000 and the input is $32$-bit integers. The running times are obtained on a 96-core machine with two-way hyper-threading, and more details, including memory profiling, are presented in Section \ref{exp section}.}
\label{run time comparison}
\end{minipage}%
}
\end{figure}
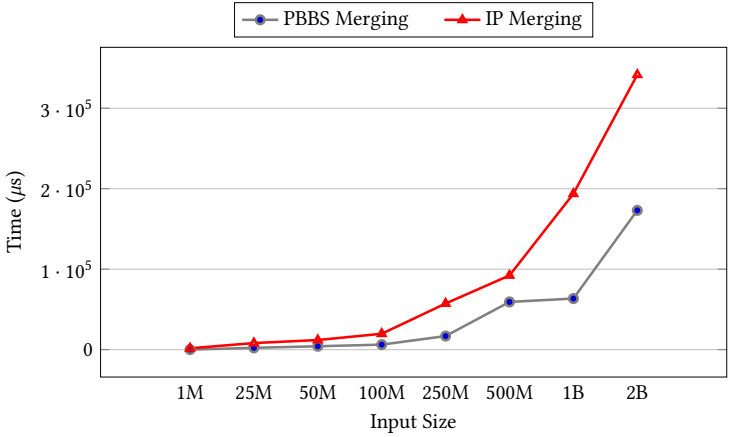

\section{Preliminaries}
    \subsection{Notation} If $A$ is an array of $n$ elements, we use $A[i:j]$ to refer to the contiguous subsequence $(A[i],A[i+1],\dots,A[j-1])$. We denote uniform sampling from a finite set $X$ as $\overset{R}{\gets} X$.
    
    \subsection{Work-Span Model} In this paper, we use the work-span model with binary forking for analyzing parallel algorithms \cite{DBLP:journals/corr/abs-1903-04650}. As such, we assume we have a set of threads with access to a shared memory. Each thread supports the same operations as in the sequential RAM model, in addition to a fork instruction, which when executed, spawns two child threads and suspends the executing (parent) thread. When the two child threads end (through executing the end instruction), the parent thread starts again by first executing a join instruction. An individual thread can allocate a fixed amount of memory private to the allocating thread (refereed to as stack-allocated memory) or shared by all threads (referred to as heap-allocated memory). The \textit{work} of an algorithm is the total number of instructions carried out in the execution of the algorithm, and the \textit{span} is the length of the longest sequence of dependent instructions in the computation. A randomized work-stealing scheduler can execute an algorithm with work $W$ and span $D$ in $W/P + O(D)$ time \whp\footnote{We say $O(f(n))$ \textbf{with high probability (\whp)} to indicate $O(cf(n))$ with probability at least $1 - n^{-c}$ for $c \geq 1$, where $n$ is the input size.} on $P$ processors \cite{doi:10.1137/S0097539793259471}. The binary forking model can be extended with additional instructions to support concurrent reads and writes in an atomic fashion by using the compare-and-swap instruction \cite{10.5555/2385452}. Generally, we only require such instructions when reads and writes on the same memory location potentially overlap in execution. Of the algorithms we present in this paper, only the algorithm for generating a random permutation, given in Section \ref{randperm}, requires the use of atomic reads and writes.

    \subsection{Sequential Stack Space} An important concept in this paper is the notion of sequential stack-allocated memory. Formally, a parallel algorithm on an input of size $n$ uses $O(f(n))$ sequential stack-allocated memory if when scheduled using a single thread, the algorithm uses at most $O(f(n))$ stack-allocated memory.
    
    In the strong PIP model, auxiliary space is measured in the sequential setting, whereas span is analyzed according to the fork-join computation graph. In this context, strong PIP algorithms are able to achieve low span and low auxiliary space simultaneously. For a more complete discussion we refer to Section 3 of \cite{2021pip}.

    \subsection{The Strong PIP Model}\label{strong PIP} The \textit{strong PIP model} assumes a fork-join computation using $O(\log n)$-word stack-allocated space sequentially. An algorithm is a \textit{strong PIP} algorithm if it runs in the strong PIP model and has polylogarthimic span \cite{2021pip}. The algorithms in this paper will be presented through the strong PIP model. However, many of algorithms we propose do not require $O(\log n)$ sequential stack memory and would, with minimal effort, fit into other models of in-place parallel computation such as the in-place PRAM model.

\section{Encoding}\label{encodingdisc}

We make an important assumption throughout the paper that the elements of the input are taken from a totally ordered universe and that comparing two elements is a constant time operation. The main benefit of this assumption is that we will be able to store information implicitly in our input using inversions. In this section, we describe how this implicit storage works and the encoding techniques used to implement it.

These ideas have been used extensively in the design of in-place algorithms and implicit data structures \cite{bronnimann2004towards, 10.1007/11785293_10, kronrod1969, IANMUNRO198666, SaloweSteiger1985, inplacemerge}. More recently, they were used by \cite{kuszmaul2020inplaceparallelpartitionalgorithmsusing} to provide an in-place parallel algorithm for the partition problem. We refer to the encoding technique we use as \textit{inversion encoding}, which uses the relative order of pairs to encode bits. In the literature, this encoding technique has sometimes been referred to as bit-stealing, but we believe inversion encoding to be a more descriptive name.

\subsection{Inversion Encoding}
To start, we fix a block $B$ of $b$ unique elements. By considering the relative ordering of consecutive pairs of elements in $B$, we can encode $b/2$ bits of additional information. Specifically, we fix the convention that a pair $(x_1,x_2)$ encodes a $0$ if $x_1 < x_2$ and a $1$ otherwise. See Figure \ref{encodingtikz} for an example.

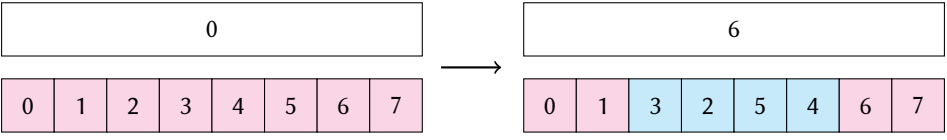
\begin{figure}[t]
\begin{tikzpicture}[
  font=\sffamily,
  box/.style={
    draw,
    minimum size=0.7cm,
    align=center,
    inner sep=0pt
  }
]

\node[box, fill = magenta!20] (A1) at (0*0.7,0) {0};
\node[box, fill = magenta!20] (A2) at (1*0.7,0) {1};
\node[box, fill = magenta!20] (A3) at (2*0.7,0) {2};
\node[box, fill = magenta!20] (A4) at (3*0.7,0) {3};
\node[box, fill = magenta!20] (A5) at (4*0.7,0) {4};
\node[box, fill = magenta!20] (A6) at (5*0.7,0) {5};
\node[box, fill = magenta!20] (A7) at (6*0.7,0) {6};
\node[box, fill = magenta!20] (A8) at (7*0.7,0) {7};
\node[box, minimum width = 8*0.7 cm]   (B1) at (7*0.7 / 2, 1)  {0};

\draw[->, thick] (8 * 0.7 - 0.1,0.5) -- (8 * 0.7 + .7,0.5);

\node[box, fill = magenta!20] (C1) at (0*0.7 + 0.7*8 + 1.35,0) {0};
\node[box, fill = magenta!20] (C2) at (1*0.7 + 0.7*8 + 1.35,0) {1};
\node[box, fill = cyan!20] (C3) at (2*0.7 + 0.7*8 + 1.35,0) {3};
\node[box, fill = cyan!20] (C4) at (3*0.7 + 0.7*8 + 1.35,0) {2};
\node[box, fill = cyan!20] (C5) at (4*0.7 + 0.7*8 + 1.35,0) {5};
\node[box, fill = cyan!20] (C6) at (5*0.7 + 0.7*8 + 1.35,0) {4};
\node[box, fill = magenta!20] (C7) at (6*0.7 + 0.7*8 + 1.35,0) {6};
\node[box, fill = magenta!20] (C8) at (7*0.7 + 0.7*8 + 1.35,0) {7};
\node[box, minimum width = 8*0.7 cm]   (B2) at (8*0.7 + 1 + 8*0.7/2, 1)  {6};

\end{tikzpicture}
\caption{An example for encoding $6$ in $B = [0,1,2,3,4,5,6,7]$.}\label{encodingtikz}
\end{figure}

We can encode and decode bit strings of size $b/2$ (equivalently integers $v$ such that $\lfloor \log_2 v \rfloor = b/2$) in $B$ using simple parallel divide-and-conquer algorithms given by \cite{kuszmaul2020inplaceparallelpartitionalgorithmsusing}. For completeness, we give them below.

{
\SetKwProg{CFn}{Function}{}{}
\setlength{\interspacetitleruled}{0pt}%
\setlength{\algotitleheightrule}{0pt}%
\begin{algorithm2e}[ht]
\DontPrintSemicolon
\SetKwBlock{DoParallel}{do in parallel}{end}
\CFn {\mf{ReadBlock}(B, i, j) \Comment{$j - i$ is assumed even}} { 
    \If{$j - i = 2$}{
        \textbf{return} $0$ if $B[i] < B[j]$ and $1$ otherwise
    }
    \DoParallel{
        $v_l \gets \mf{ReadBlock}(B,i,i+(j-i)/2)$ \\
        $v_r \gets \mf{ReadBlock}(B,i+(j-i)/2 + 1, j)$
    }
    \textbf{return} $v_r \cdot 2^{\frac{j-i}{4}} + v_l$
}
\end{algorithm2e}
}
{
\SetKwProg{CFn}{Function}{}{}
\setlength{\interspacetitleruled}{0pt}%
\setlength{\algotitleheightrule}{0pt}%
\begin{algorithm2e}[ht]
\SetKwFor{ParFor}{parallel for}{do}{endfch}
\DontPrintSemicolon
\CFn {\mf{WriteBlock}(B, i, j, v) \Comment{$j - i$ is assumed even and equal to $2\lfloor \log_2 v \rfloor$}} { 
    \ParFor{$k \gets 0$ to $(j-i)/2$}{
        \If{$k$th bit of $v$ is a $0$}{
            $t \gets B[2k]$ \\
            $B[2k] \gets 
            \min(B[2k],B[2k+1])$ \\
            $B[2k+1] \gets \max(t,B[2k+1])$
        } \Else {
            $t \gets B[2k]$ \\
            $B[2k] \gets 
            \max(B[2k],B[2k+1])$ \\
            $B[2k+1] \gets \min(t,B[2k+1])$
        }
    }
}
\end{algorithm2e}
}

With these functions, we can, in a block of size $b = 2s\log n + 2t$, encode $s$ $(\log n)$-bit numbers and $t$ bits. Reading and writing $B[i:j]$ requires $O(\log |i-j|) = O(\log n)$ work and $O(\log\log |i - j|) = O(\log\log n)$ span. Note that both ReadBlock and WriteBlock require only a constant amount of sequentially allocated stack space and are strong PIP algorithms. 



\subsection{Encoding Notation} Let $A$ be an array of $n$ unique elements and $b$ a block size. Let $X_b$ denote an array of $b/2$-bit elements encoded in $A$ using inversion encoding. Then $X_b[i]$ represents the $i$th value of $X_b$ encoded in the $i$th block of $A$. Further, reads and writes on $X_b[i]$ should be interpreted as ReadBlock and WriteBlock invocations on the $i$th block of $A$. That is, $X_b[i] \gets j $ is interpreted as WriteBlock$(A,ib, (i+1)b ),j)$ and $t \gets X_b[i]$ is interpreted as $t \gets $ ReadBlock$(A,ib, (i+1)b)$ .

\subsection{Restorable Buffers}
\label{restorable-buffers} A central construction that we will extensively use is what we call a \textit{restorable buffer}. The purpose of this buffer is to obtain additional space that can be treated as `true' auxiliary space. With this auxiliary space, we can use existing non-in-place algorithms. We make the following definition.

\begin{defn}\label{buffer-defn}
    A restorable buffer $B(s,w,l)$, parameterized by an auxiliary length $s$, a word size $w$, and a pointer $l$, is a contiguous sequence of memory starting at $l$ of size $s+2ws$. The first $s$ elements starting at $l$ are called the auxiliary elements, and the last $2ws$ elements are called encoding elements. The first range is called the auxiliary part and the second range is called the encoding part.
\end{defn}

Any buffer $B(s,w,l)$ supports initialization and restoration operations. Initialization encodes the first $w$ bits of elements occupying the auxiliary part into the encoded part, and restoration decodes the encoded part into the first $w$ bits of each element in the auxiliary part. Both of these operations are done in parallel using the read and write block operations of the previous section.

We now discuss an additional functionality we may add to a buffer that will be useful later when we give an algorithm for random permutation.
\begin{defn}\label{simulated-buffer-defn}
    An adjustable buffer is a restorable buffer $B(s,w,l)$ which supports the following additional operations.
    \begin{itemize}
        \item Simulated auxiliary operations: The elements that occupied the auxiliary part of the buffer are always accessible. We may read these values implicitly by decoding the corresponding $w$ bits and reading all but the first $w$ bits of the element occupying that space. We can write these values in the same way.
        \item Encoding element operations: the elements that are being used for encoding can still be read and written without issue. However, writing a value can potentially break the bit currently being encoded. As such, we support write operations, but enforce that after any such operation, the corresponding pair is swapped so that the encoding bit is not destroyed.
    \end{itemize}
\end{defn}

These two operations can face concurrency issues. As such, we enforce that the simulated auxiliary operations can never occur concurrently with writes to encoding elements. Additionally, we enforce that two members of the same encoding pair can never be written to concurrently. Here we remark that adjustable buffers only exist if we have the ability the modify the underlying bit representation of the input.

\subsection{Duplicate Elements}

Thus far, the encoding schemes presented assume that inputs do not contain duplicate elements. To remove this assumption, we may use the construction given in Section 3 of \cite{kuszmaul2020inplaceparallelpartitionalgorithmsusing}. With this construction in a block $B$ of $b$ elements, we can encode at least $b/4$ bits. Handling duplicates does not materially change any of the algorithms we give. As such, we will assume throughout the rest of the paper that all elements of the input are distinct, in order to ease the presentation.


\subsection{Obstacles to Inversion Encoding} There are two mains obstacles to using inversion encoding in addition to cost of reading and writing encoded values. The first obstacle is that if an algorithm modifies elements of the input, then this can destroy the bits encoded in the corresponding pairs. If the input is static, meaning the algorithm never directly modifies it, this obstacle is avoided. Otherwise, a more careful approach is required.



The second issue that arises is that an algorithm may need to read the original state of the input. Since encoding values swaps elements in pairs, unless we know the initial state of every pair, it can be impossible to correctly read the original state of the input. We can avoid this obstacle if the input has a known initial structure, e.g. is sorted.

\section{Merging}\label{merge section}

\subsection{Introduction}
Merging is a fundamental primitive in algorithm design. It has been extensively studied in both the sequential and parallel setting \cite{88483,2021pip}. Further, there is a large body of work devoted to the study of in-place merge algorithms \cite{inplacemerge, kronrod1969, 10.1007/11785293_10,SaloweSteiger1985}. To the best of our knowledge, there does not exist a parallel algorithm for in-place merging that is able to achieve both low span and optimal work simultaneously. In this section, we bridge this gap by giving a work-efficient near-optimal span strong PIP algorithm for merging two sorted arrays. Our main result is the following.

\begin{thm}\label{ipmerge}
   Merging two sorted arrays of size $n$ and $m$ (where $n \geq m$) can be done with $O(n)$ work, $O(\log^2 n)$ span whp, and using $O(\log n)$ sequential stack-allocated memory.
\end{thm}

Let $A$ and $B$ be sorted arrays of sizes $n$ and $m$. For simplicity, we assume $n$ and $m$ are each divisible by a parameter $b$; in Section \ref{removing divisibility} we remove this assumption. We view $A$ and $B$ as partitioned into blocks of size $b$. Hence $A$ becomes the sequence $X_A$ (with $n/b$ blocks) and $B$ becomes the sequence $X_B$ (with $m/b$ blocks). Let $N = \frac{n}{b} + \frac{m}{b}$. We define $X$ as the sequence of $N$ blocks given by:
\[X[i] = \begin{cases}
    X_A[i] & \text{ if } 0 \leq i < \frac{n}{b} \\
    X_B[i - \frac{n}{b}] & \text{ otherwise.}
\end{cases}\]

Our merging algorithm has two main phases.
\begin{enumerate}[1.]
    \item \textbf{The transformation phase}. In this phase, we merge in-place the sequences of blocks $X_A$ and $X_B$ by the block endpoints. To facilitate this process, we use the encoding techniques discussed in Section \ref{encodingdisc} to equip every block $X[i]$ with the following data.
    \begin{itemize}
        \renewcommand\labelitemi{--}
        \item The block's rank according to endpoint in the opposite block sequence.
        \item The block's final position after merging by block endpoints.
        \item An inversion pointer identifying the sole successive block (if any) with which it can have an inversion after merging.
    \end{itemize}
    
    \item \textbf{The sort phase}. After the transformation phase, each block can have an inversion with at most one successive block. We perform a divide-and-conquer procedure to eliminate these inversions, ultimately producing a globally sorted array.
\end{enumerate}

Throughout this process, we only use $O(\log n)$ sequential-stack allocated memory and the total work is $O(n)$. In the next few sections, we detail each phase of the algorithm, and in Section \ref{final rem}, we discuss how to handle duplicate elements, reduce the sequential-stack memory to $O(1)$, and remove the dependence on randomness.

\subsection{Transformation Phase}\label{tphasemerge} We describe how to merge $X_A$ and $X_B$ by endpoints, and compute and encode inversion pointers. Let
\[b = 8\log(n+m) + 7.\]
Each block $X[i]$ has size $b$. Within $X[i]$, we encode $4$ $\log (n+m)$-bit numbers and $3$ bits:
\[
\begin{array}{ll}
T_b[i] &  \text{The end-sorted position of } X[i]; \\
\inv_b[i] & \text{The pointer identifying which block can form an inversion with } X[i]; \\
R_b[i] &  \text{The rank of } X[i] \text{ by endpoint among blocks in the opposite sequence}; \\
I_b[i] &  \text{The index of } X[i]; \\
C_b[i], D_b[i], E_b[i] &  \text{A random bit, a done flag, and a swap flag respectively}.
\end{array}
\]
Note the last element of each block is not used in the encoding scheme to avoid issues of concurrency (we will be binary searching concurrently over endpoints of blocks). In addition, we split every block $X[i]$ into $X_L[i]$ and $X_R[i]$
where
\[X_L[i] = X[i][0 : s] \text{ and } X_R[i] = X[i][s : b] \text{ for } s = 6\log(n+m) + 2.\]
Here $X_L[i]$ (the `left side') holds most of the encoded data, while $X_R[i]$ (the `right side') contains $T_b[i], C_b[i]$, and $D_b[i]$. This separation is used to prevent conflicts that arise when concurrently reading encoded values during block swaps. We detail this further in Section \ref{endmergesection}.

\subsubsection{Computing ranks, end-sorted positions, and inversion pointers}\label{compRandI} First, we compute the rank according to endpoint of each block $X[i]$ in the opposite sequence via parallel binary searches. Specifically, if $X[i]$ belongs to $X_A$, we binary-search for its endpoint over the endpoints of $X_B$ and if $X[i]$ belongs to $X_B$, we do the reverse. The resulting rank of $X[i]$ (encoded in $R_b[i]$) determines its end-sorted block position
\[T_b[i] = i + R_b[i],\]
which encodes where $X[i]$ appears among all $N$ blocks once merged by endpoints. Next, we compute the inversion pointer $\inv_b[i]$. After merging by endpoints, each block $X[i]$ can potentially contain an inversion with exactly one other block, namely the block in the opposite sequence with the smallest endpoint larger than $X[i]$'s endpoint. After merging, this block has index
\[\inv_b[i] = T_b[R_b[i]].\]
We encode the values $R_b[i]$ and $\inv_b[i]$ into the left side of each block $X[i]$ and $T_b[i]$ into the right side of each block $X[i]$.

\subsubsection{End-merging}\label{endmergesection} We now give an in-place algorithm to merge $X_A$ and $X_B$ by endpoints, with the values $R_b[i],$ $T_b[i],$ and $\inv_b[i]$ already computed. At a high-level, the algorithm employs a standard symmetry breaking technique to insert each block into its end-sorted position.

Each block $X[i]$ has its desired end-sorted position $T_b[i]$, and we will refer to $X[T_b[i]]$ as the target of $X[i]$. To move every block to its final position, we iterate rounds wherein each block that is not in its desired end-sorted position flips a fair coin. Ideally, we have each block $X[i]$ swap with block $X[T_b[i]]$ if $X[i]$ flips heads ($C_b[i] = 1$) and $X[T_b[i]]$ flips tails ($C_b[T_b[i]] = 0$). Then in each round, a fraction of the remaining blocks are expected to perform swaps so that after $O(\log (n+m))$ rounds \whp, all blocks will be in their correct end-sorted position. However, a subtle issue arises due to the nature of inversion encoding; when a block $X[i]$ reads its encoded values to determine whether to swap with its target, it can happen that another block $X[j]$ whose target is $X[i]$ initiates a swap with $X[i]$ concurrently, thereby changing the encoded values read from $X[i]$. To ensure blocks are processed consistently, we will swap blocks in two stages, where the first stage will read encoded data from right sides and swap left sides, and the second stage will do the reverse.

We now describe our EndMerge algorithm in detail. The pseudocode appears in Algorithms \ref{EndMerge} and \ref{done-alg}. To start, in parallel, we initialize each block $X[i]$'s swap flag ($E_b[i]$) to false (line 4). Additionally, if $X[i]$ is already in its end-sorted position, we set the done flag ($D_b[i]$) to true (line 5). Otherwise, we initialize the done and swap flags to false (line 7). This concludes the initialization phase of the algorithm. Going forward, we say that a block is \textit{active} in an iteration if it is not in its end-sorted position. We now iterate rounds consisting of three main steps:

\begin{enumerate}[1.]
    \item \textbf{Coin flipping: } In parallel for each block $X[i]$, if $X[i]$ is active, we flip a random coin and encode the result in $C_b[i]$ (lines 10-11).
    \item \textbf{Left swaps: } At this point, each active block $X[i]$ has its end-sorted position, its done flag, and its coin flip encoded in its right side $X_R[i]$. Reading from $X_R[i]$ in parallel, we check if $X[i]$ is done, has flipped heads, and its target has flipped tails (line 13). If these conditions are met, $X[i]$ sets its done flag to true, its swap flag to true, and writes its own index $i$ into $I_b[i]$ (line 14). Finally, we swap the left sides $X_L[i]$ and $X_L[T_b[i]]$ (line 15). We note that since only the left side of active blocks are swapped in this step, it is safe to read the values $T_b[i], C_b[i],$ and $D_b[i]$ encoded in the right side of each block.
    \item \textbf{Right swaps: } We complete the block swaps initiated in the previous step. Reading from $X_L[i]$ in parallel, we check if the swap flag is true (line 17). If the swap flag is true, this means the left side $X_L[i]$ was part of a swap in the previous step. The index from which that swap was made is accessible by $I_b[i]$. We thus complete the block swap by swapping $X_R[I_b[i]]$ with $X_R[i]$ (line 20). Before doing this swap, we also check if $X[i]$ and $X[I_b[i]]$ form a two-cycle in terms of their end-sorted positions. In this case, we set the done flag of $X[I_b[i]]$ to true as well (line 19).
\end{enumerate}

To determine when the iteration should stop we use a concurrent bit in shared memory. This bit represents a flag and is initialized to true. Then at the beginning of each round, in parallel for each block $X[i]$, we check if the block is not in its desired end-sorted position (lines 5-6 of the Done function), and if so set the flag to false. The flag bit can be reused in each round and thus only contributes one bit of shared memory.

\begin{algorithm2e}[ht] \caption{EndMerge$(X)$} \label{EndMerge}
\DontPrintSemicolon
\SetKwFor{ParForEach}{parallel foreach}{do}{endfch}
\textbf{Input:} Sequence of $N$ blocks $X$ of size $b$ \Comment{each block in $X$ encodes the data mentioned above} \\
\textbf{Output:} $X$ sorted by endpoints \\
\ParForEach{block $X[i]$}{
    $E_b[i] \gets 0$ \\
    \lIf{$T_b[i] = i$}{$D_b[i] \gets 1$} 
    \Else{
        $D_b[i] \gets 0$
    }
}
\While{not Done$(X)$} {
    \ParForEach{block $X[i]$}{
        \If(\Comment{block $X[i]$ is active}){$D_b[i] = 0$}{
            $C_b[i] \overset{R}{\gets} \{0,1\}$
        }
    }
    \ParForEach{block $X[i]$}{
        \If{$D_b[i] = 0$ and $C_b[i] = 1$ and $C_b[T_b[i]] = 0$} {
            $D_b[i] \gets 1, \ E_b[i] \gets 1, \ I_b[i] \gets i$ \\
            swap$(X_L[i], X_L[T_b[i]])$
        }
    }
    
    \ParForEach{block $X[i]$}{
        \If{$E_b[i] = 1$} {
            $D_b[i] \gets 1, \ E_b[i] \gets 0$ \\
            \lIf{$T_b[i] = I_b[i]$}{$D_b[I_b[i]] \gets 1$}
            swap$(X_R[I_b[i]],X_R[i])$
        }
    }
}
\end{algorithm2e}

\begin{algorithm2e}[ht] \caption{Done}\label{done-alg}
\DontPrintSemicolon
\SetKwFor{ParForEach}{parallel foreach}{do}{endfch}
\textbf{Input:} Sequence of blocks $X$ \Comment{each block in $X$ encodes the data mentioned above} \\
\textbf{Output:} true if $X$ is end-sorted, false otherwise \\
flag $\gets$ 1 \Comment{stored in the heap} \\
\ParForEach{block $X[i]$}{
    \lIf{$D_b[i] = 0$}{flag $\gets$ 0}
}
\textbf{return} flag
\end{algorithm2e}

We now analyze the complexity of the transformation phase. Computing ranks, end-sorted positions, and inversion pointers requires $O(n)$ work and $O(\log n)$ span. We now show that the work and span of EndMerge are $O(n)$ and $O(\log^2 n)$ \whp. To start, the initialization phase of the algorithm consists of a single parallel for loop over each block. Here, the dominating cost is the reading and writing of $T_b[i]$ for each $X[i]$. Since both of these values are encoded using $O(\log n)$ pairs, the ReadBlock and WriteBlock functions use $O(\log n)$ work and span. Thus the total work for this phase is $O(N \log n) = O(n)$ and the span is $O(\log N \log n) = O(\log^2 n)$.

Next, we analyze the while loop in EndMerge. Let $R$ be the number of rounds for which the while-loop runs. Then $R$ is simply the longest number of iterations a block remains active. We claim that $R = O(\log n)$ \whp. In each iteration, an active block has at least a 1/4 chance of being placed in its desired end-sorted position. Thus the number of iterations that block $X[i]$ remains active is stochastically dominated by a geometric random variable $Y_i$ with probability 1/4. This means that $R$ is stochastically dominated by $\max_i Y_i$, and a basic union bound establishes that $R \leq 2\log n$ \whp.

The span contributed by the while loop is easily seen to be $O(R \log n) = O(\log^2 n)$, as the four parallel loops contribute $O(\log n)$ span per round. The work is more complicated. First, note that each block trivially contributes $O(1)$ work per round, whereas active blocks potentially contribute $O(\log n)$ work in reading $C_b[T_b[i]]$. Thus the work contributed by inactive blocks is $O(R N) = O(\log n N) = O(n)$. We are reduced to considering the work done by active blocks. 

\begin{lem}
    The work $W$ contributed by active blocks is $O(n)$ \whp.
\end{lem}
\begin{proof}
Let $G_i = m$ be the number of active blocks at the start of iteration $i$, with active block index set $L$. Moreover, assume that $m \geq \sqrt{n}$. Let $Z_{j}$ be the indicator for $X[j]$ becoming inactive at stage $i+1$. Let $Z = \sum_{j \in L} Z_{j}$. Then $G_{i+1} = G_i - Z$. There is a subset $L'$ of $L$ of size $m/3$, such that $Z_j$ and $Z_k$ are independent for $j,k \in L'$. This is because $Z_j$ and $Z_k$ are independent if and only if $\{j, T_b[j]\}$ and $\{k, T_b[k]\}$ are disjoint, so we can exhaustively choose enough indices from $L$ to construct a large mutually independent set $L'$. Say their indicator sum is $Z'$. The Chernoff bound for Bernoulli variables says that
    \[\Pr[Z' \leq \mu/2] < \exp(-\mu/8), \ \mu = \E[Z'].\]
    Note that $\mu \geq m/12$ since $\Pr[B_j = 1] \geq 1/4$, for any $j$. Thus 
    \[
       \Pr[Z' \leq m/24] < \Pr[Z' \leq \mu/2] < \exp(-\mu/8) < \exp(-m/192) < \exp(-\sqrt{n}/192) < 1/n^2.
    \]
    Note that $23 m/24 \leq G_{i+1}$ implies that $23m/24 \leq m-Z$ which further means $Z' \leq Z \leq m/24$. Thus $\Pr[23 m/24 \leq G_{i+1}] \leq Pr[Z' \leq m/24] < 1/n^2$. We then conclude that $\Pr[23 G_i / 24 \leq G_{i+1}] < 1/n^2$, assuming $G_i \geq \sqrt{n}$.

    Let $I = \{i : G_i \geq \sqrt{n}\}$ and $J = \{j : \sqrt{n} > G_j \geq 1\}$. Note $|I| + |J| = R \geq \log n$ \whp, so assume it is the case. Then, in particular $|I| < \log n$. Hence, by a union bound, $\Pr[\exists i \in I ( 23 G_i / 24 \leq G_{i+1}) ] \leq \log n/n^2 < 1/n$. Thus assume $23 G_i / 24 \leq G_{i+1}$ for all $i \in I$. Then $\sum_{i \in I} G_i \leq n \sum_i (23/24)^i = O(n)$. Moreover, $\sum_{i \in J} G_i \leq |J|\sqrt{n} \leq \sqrt{n}\log n < n$. Hence $W = O(n)$ \whp.
\end{proof}

Finally, we note that each part of the transformation phase is easily seen to use $O(1)$ sequential stack space. Combining everything, we conclude the following lemma.

\begin{lem}\label{merge-transformation-phase}
    The transformation phase takes $O(n)$ work and $O(\log^2 n)$ span whp and uses only $O(\log n)$ sequential stack-allocated memory.
\end{lem}

\subsection{Sort Phase}

Before describing the sort phase, we make some definitions and an observation. Let $S$ be a block sequence. We say that $S[i] \leq S[j]$ if there is no element of $S[i]$ inverted with an element of $S[j]$. Call a block $S[i]$ \emph{in-order} if $i < j$ implies $S[i] \leq S[j]$, and \emph{almost in-order} if there is at most one $j > i$ such that $S[i] \not\leq S[j]$. Observe that a block sequence being sorted is equivalent to every block being in-order and every block being sorted.

As noted in Section \ref{compRandI}, after end-merging, every block in $X$ is sorted and almost in-order, and moreover, if any block is out of order, then the block that witnesses this is $X[\inv_b[i]]$. We now describe a top-down divide-and-conquer algorithm that sorts $X$ in-place. At the top level we split $X$ in half (by blocks). It follows from previous discussion that, with respect to the splitting point, the only potential inversions exist between the rightmost block on the left half, say $X[m]$, and the block on the right half whose index is given by $j = \inv_b[m]$. Then after merging $X[m]$ and $X[j]$ out-of-place in the stack, the left and right sub-problems are effectively disjoint and can be processed in parallel. We note that after merging $X[m]$ and $X[j]$, the inversion pointers need to be updated. The pointers on the right remain the same, but inversion pointers from the left, which point past the midpoint, should now point to the end of the left slice. Rather than update every inversion pointer, we simply take the inversion pointer of $X[i]$ to be the minimum between the initial inversion pointer $\inv_b[i]$ and the end of the current slice (line 2 in Separate). We give the pseudocode in Algorithms \ref{separatealg} and \ref{seqsortalg}. 

\begin{algorithm2e}[ht]
\SetKwFor{ParFor}{parallel for}{do}{endfch}
\caption{Separate}\label{separatealg}
\textbf{Input:} A block sequence $X$ with $K$ blocks of size $b$, $\inv_b$ is the initial data of the inversion pointers \\

$i \gets K/2-1$ and $j \gets \min(K-1,\inv_b[i])$ \\
$n_i \gets \inv_b[i]$ and $n_j \gets \inv_b[j]$ \Comment{store encoded data in stack}\\
$A \gets X[i]$ and $B \gets X[j]$ \\
reset the encoding in $A$ and $B$ \\
$E \gets $ merge$(A,B)$ \\
$X[i] \gets E[0:b]$ and $X[j] \gets E[b:2b]$ \\
$\inv_b[i] \gets n_i$ and $\inv_b[j] \gets n_j$ \Comment{encode inversions pointers back into $A$ and $B$}
\end{algorithm2e}

\begin{algorithm2e} [ht] 
\DontPrintSemicolon
\SetKwBlock{DoParallel}{do in parallel}{end}
\caption{SeqSort}\label{seqsortalg}
\textbf{Input:} A block sequence $X$ with $K$ blocks of size $b$. \\
\textbf{Output:} $X$ in sorted order \\
\lIf{$K$ is 1}{
    reset the encoding and return
}
Separate($X$) \\
\DoParallel{
    SeqSort($X[0:K/2])$ \\
    SeqSort($X[K/2:K])$
}
\end{algorithm2e}

Now we show that the algorithm is correct. We prove by induction that at any point in the recursion, for the recursion slice $S = X[a:b]$, the following holds: for each $a \leq i \leq b$, if $i < j$ and $X[i] \not\leq X[j]$ then $j = \min(b-1,\inv_b[i])$. We denote this statement by $(*_S)$. Note the base case of this induction is exactly the observation we made prior about the block sequence $X$ after the transformation phase. We now prove the inductive step.

Suppose $(*_S)$ is true for a recursion slice $S = X[a:b]$. Let $c = (a+b)/2 - 1$, and let $w(i) = \min(b-1,\inv_b[i])$. The inductive hypothesis asserts that if $X[i]$ is part of an inversion, the out of order block is given by $X[w(i)]$.

Let $X'[a : b]$ be the array obtained from $X[a : b]$ after calling Separate. Then $X'_L = X'[a:c+1]$ and $X'_R = X'[c+1:b]$ are the two recursion slices at the next depth. Let $w_L(i) = \min(c,\inv_b[i])$ and $w_R(i) = \min(b-1,\inv_b[i])$. We will show $(*_{X'_L})$ and $(*_{X'_R})$. First we establish some useful relations, where we let $B$ denote an arbitrary block of $X$.
\begin{equation}\label{rel1}
    i \neq c,w(c) \implies X'[i] = X[i]
\end{equation}
\begin{equation}\label{rel2}
    i \leq c \implies X'[i] \leq X'[w(c)]
\end{equation}
\begin{equation}\label{rel3}
    X[c] \leq B \implies X'[c] \leq B
\end{equation}
\begin{equation}\label{rel5}
    B \leq X[c], X[w(c)] \implies B \leq X'[c]
\end{equation}
\begin{equation}\label{rel7}
    B \geq X[c], X[w(c)] \implies B \geq X'[w(c)]
\end{equation}
All of these relations ultimately reduce to the fact that $X'[c]$ and $X'[w(c)]$ contain the smallest and largest elements from $X[c] \cup X[w(c)]$ respectively. Also note, using (\ref{rel1}) and the inductive hypothesis, that if $i \neq c,w(c)$ and $j \neq w(i), w(c), c$ then
\begin{equation}\label{easy-case}
    X'[i] = X[i] \leq X[j] = X'[j].
\end{equation}

We first show $(*_{X'_L})$. Fix $a \leq i \leq c$. We want to show that if $j > i$ and $j \neq w_L(i)$ then $X'[i] \leq X'[j]$. By (\ref{easy-case}), for $i < c$, it always suffices to consider only $j \in \{w(i),w(c),c\}$, so we treat this case first.
This further splits into three cases.
\begin{description}
   \item[Case 1. $w(i) < c \colon $]
   In this case $w_L(i) = w(i)$. As $j \neq w_L(i)$, we have $j = w(c)$ or $c$. Notice that $w(i) < c$ implies that $X[i] \leq X[c],X[w(c)]$, so $X'[i] = X[i] \leq X'[c] \leq X'[w(c)]$ by (\ref{rel5}), (\ref{rel1}), (\ref{rel2}), and we are done.
   
   \item[Case 2. $w(i) = c \colon $]
   Here $w_L(i) = c$. We thus have $j = w(c)$, so $X'[i] \leq X'[w(c)]$ by (\ref{rel2}).

   \item[Case 3. $w(i) > c \colon $]
   Here $w_L(i) = c$. This means that $i < c < w(i)$, which means that the block potentially out of order with $X[c]$ is $X[w(i)]$, so $w(i) = w(c)$ by the inductive hypothesis.
   As $j \neq w_L(i)$, we have $j = w(c) = w(i)$. But then we are done by (\ref{rel2}).
\end{description}

To finish the proof of $(*_{X'_L})$, we now only need to consider $i = c$. In this case $w_L(c) = c$. For any $j > c$ and $j \neq w(c)$, we have $X[c] \leq X[j] = X'[j]$, so (\ref{rel3}) gives $X'[c] \leq X'[j]$. If $j = w(c)$, then $X'[c] \leq X'[w(c)]$ by (\ref{rel2}). Hence we conclude $(*_{X'_L})$.

Now we show $(*_{X'_R})$. First note that $w_R(i) = w(i)$. Let $c+1 \leq i \leq b$. We want to show that if $j > i$ and $j \neq w(i)$, then $X'[i] \leq X'[j]$. We split into two cases.
\begin{description}
   \item[Case 1. $i \neq w(c)\colon $]
   By (\ref{easy-case}), it suffices to consider $j =w(c)$, with $w(i) \neq w(c)$. 
   Since $w(i) \neq w(c)$, then $X'[i] = X[i] \leq X[w(c)]$ so $X'[i] \leq X'[w(c)]$ by (\ref{rel3}). 
   
   \item[Case 2. $i = w(c)\colon $]
   Let $j > w(c)$ with $j \neq w(w(c))$. Then $X[c],X[w(c)] \leq X[j]$ and thus $X'[w(c)] \leq X'[j]$ by (\ref{rel7}), which is what we desire.
\end{description}
We conclude $(*_{X'_R})$.

This induction implies that SeqSort is correct, because at the base case of the recursion (when the recursion slices are size 1), the inductive hypothesis simply asserts that the block is in-order. Moreover, the blocks start sorted (ignoring the encoding which is reverted at the end of the algorithm), and this is preserved by the algorithm, so at the base case all the blocks are sorted as well. Together with the in-order property, this implies that the final array is sorted.

We analyze the complexity of SeqSort. The work follows the recurrence 
\[W(N) = 2W(N/2) + O(\log n),\]
so $W(N)= O(N b)$.
The span satisfies
\[D(N) = D(N/2) + O(\log n)\]
assuming the merge in the Separate subroutine is done sequentially in the stack, so $D(N) = O(\log N \log n)$. Finally, the only non-trivial auxiliary memory required during the Sort phase is the $O(\log n)$ sequential stack-allocated memory used to merge blocks $A$ and $B$ out of place in the Separate subroutine (lines 4-5). We conclude the following lemma.

\begin{lem}\label{merge-sort-phase}
    The sort phase takes $O(n)$ work and $O(\log^2 n)$ span and uses only $O(\log n)$ sequential stack-allocated memory.
\end{lem}

\subsection{Removing the Divisibility Assumption} \label{removing divisibility}
We now remove the assumption of the size of the input arrays $A$ and $B$ being divisible by $b$. For convenience, we will treat $A$ and $B$ as being stored consecutively in memory, written $AB$, though this does not functionally change anything. We first write $AB$ as $A_1A_2B_1B_2$ where $A_1$ is the largest prefix of $A$ divisible by $b$ and similarly for $B_1$. We rotate the array so that it is of the form $A_1B_1A_2B_2$. Since $|A_1|$ and $|B_1|$ are divisible by $b$, we may merge $A_1B_1$ using the parallel in-place algorithm to get a sorted array $S$. We then merge $A_2B_2$ sequentially using any quadratic in-place sequential sorting algorithm (e.g. Bubble Sort) to get $T$, which has size $O(\log n)$. Our array has now become become $ST$, which we rotate to be $TS$. 

Next, we split $S$ into two halves, so that our array is of the form $T S_1 S_2$. We first restrict our attention to merging $TS_1$. Notice that $|S_2| \geq |T| \log |S|$ since $|T|$ is $O(\log n) = O(\log|S|)$, so we may encode $|T|$ values between $0$ and $|S_1|-1$ in $S_2$. In parallel, for each element of $T$, we binary search for its position in $S_1$ and use inversion encoding to encode these pointers in $S_2$. Write $T$ as $t_1t_2\dots t_l$ and $S_1$ as $\bar s_1 \bar s_2 \dots \bar s_{l+1}$ where $\bar s_i$ is the range of elements of $S_1$ between $t_{i-1}$ and $t_{i}$. These ranges are easily determined by reading the encoded pointers from $S_2$. Now, rotate $\bar s_1$ to be in front of $T$, so we have $\bar s_1 t_1 \dots t_l \bar s_2 \dots \bar s_{l+1}$. Then rotate $\bar s_2$ to get $\bar s_1 t_1 \bar s_2 t_2 \dots$. We iterate this procedure sequentially until we obtain the array $M = \bar s_1 t_1 \dots \bar s_i t_i \dots \bar s_l t_l \bar s_{l+1}$. Then $M$ is the array obtained by merging $TS_1$. We reset the encoding on $S_2$ so that it is once again sorted. We are thus reduced to merging $MS_2$.

Observe that the only elements of $M$ that can be inverted with elements of $S_2$ are elements from $T$ which are larger than every element in $S_1$. As there can be at most $|T|$ of these, the prefix of $M$ consisting of the first $|S_1|$ elements, call it $L$, is in order with respect to the remaining elements of $M$ and $S_2$. If we write $M$ as $LR$, the entire array is $LRS_2$, it thus suffices to merge $RS_2$. We may now perform the same process as before to merge $R$ and $S_2$, but instead use inversion encoding on $L$ to encode the pointers. The final result is the merged array of $S$ and $T$.

We analyze the complexity of the first step, as the second step is essentially identical. The rotation of $\bar s_i$ into place rotates $|T|+|\bar s_i|$ elements and performs two ReadBlock operations. Hence the total work is $\sum_{i = 1}^{|T|} |T|\log n + |\bar s_i| = |T|^2 \log n + |S_1| = O(n)$. The span is $O(\log^2 n)$ since the binary searches and encoding take $O(\log n)$ span, each rotation is $O(\log n)$ span, the read operation is $O(\log \log n)$ span, and there are $\log n$ rotation steps. There is an additional cost to merge $A_2B_2$, but since $|A_2|,|B_2| = O(\log n)$, it is negligible.

Putting the results of this section together with Lemma \ref{merge-transformation-phase} and Lemma \ref{merge-sort-phase}, we conclude Theorem \ref{ipmerge}.

\subsection{Final Remarks}\label{final rem}

We have produced a work-efficient near-optimal span strong PIP algorithm for merging two sorted arrays. Here, we make a couple final remarks.

\begin{enumerate}[(i)]
    \item The only meaningful stack-allocated space our algorithm uses is in the merging of two blocks of size $O(\log n)$, during the Separate procedure. We can instead use a more complicated sequential in-place merge algorithm \cite{inplacemerge} to eliminate this usage. This modification results in a Strong PIP merge algorithm using constant sequential stack space.

    \item Ultimately, the algorithm we give relies on some randomness. This only happens in the transformation phase, where we make use of a randomized symmetry breaking strategy that simplifies the situation at the cost of losing deterministic guarantees on the runtime. We can instead compute a 2-ruling set using the algorithm given in \cite{Vishkin2008ThinkingIP}, which allows us to do symmetry breaking in a deterministic fashion. This does not affect the complexity of the algorithm.

\end{enumerate}
\section{Random Permutation}\label{randperm}

\subsection{Introduction}
\subsubsection{Main result.} In this section, we consider the random permutation (or random shuffling) problem: given an array $A$, randomly permute the elements of $A$. This problem has been studied extensively in the sequential setting, the non-in-place parallel setting \cite{highly-parallel}, and the in-place parallel setting \cite{2021pip, penschuck2023engineeringsharedmemoryparallelshuffling}. Our contribution is a substantial improvement of the main result of \cite{penschuck2023engineeringsharedmemoryparallelshuffling}, as well as the results in \cite{2021pip}, under the assumption that the elements of $A$ are taken from a totally ordered universe and that we have the ability to modify the underlying bit representation of these elements. We show the following theorem.

\begin{thm}\label{rand-perm-main-thm}
    Uniformly randomly shuffling an array of size $n$ can be done with $O(n)$ expected work, $O(\poly n)$ span \whp, and using $O(1)$ sequential stack-allocated memory.
\end{thm}

\subsubsection{History.} 
We first recount the sequential Knuth shuffle algorithm \cite{durstenfeld-shuffle}. The algorithm builds a random permutation by iterating from the end of the array, and swaps the element at index $i$ with the element at index $j$, where $j$ is chosen uniformly between $1$ and $i$.


In \cite{highly-parallel}, it was shown that this algorithm is highly parallel by using reservations to perform swaps in a conflict free manner. In the recent work \cite{2021pip}, the authors observe that this parallelized version of Knuth-shuffle satisfies the so called `decomposable property'. In light of this, they give an algorithm that randomly shuffles a given array of size $n$ by iteratively processing chunks of $n^{1 - \epsilon}$ swaps in parallel. The result is a work-efficient relaxed PIP algorithm for random permutation. Importantly, while the algorithm achieves work-efficiency, both the auxiliary memory usage $O(n^{1-\epsilon})$ and the span $O(n^\epsilon \log^2 n)$ are poor. The algorithm we propose improves on these aspects by employing the buffer techniques given in Section \ref{restorable-buffers}.

\subsubsection{Overview.} At a high level, our strategy is to set up buffers to support the decomposability of parallel Knuth-shuffle; this allows us to process much larger chunks than the relaxed PIP algorithm, which in turn leads to a polylogarithmic depth algorithm. In reality, it is not so simple, and the aforementioned procedure only partially works. Instead, the algorithm operates in two stages.

\textit{Stage 1.} In a suffix of the array, we set up a restorable buffer. Using the auxiliary from the buffer, we perform a modification of the relaxed PIP algorithm mentioned previously to shuffle the entire array excluding the buffer. Restoring the buffer then concludes stage 1.

\textit{Stage 2.} In this stage, we set up an adjustable buffer in the prefix of the random permutation already built. This buffer supports simulating reads and writes on the underlying input, so we may process the remaining elements as in stage 1, where we perform simulated swaps when interacting with the buffer. After all elements have been processed, we restore the buffers and randomly swap each of the pairs that are used for encoding.

It is not obvious why this results in a correct algorithm. In using the buffer, we perform swaps on the input, thereby destroying the uniformly random shuffle built in the previous stage. The key insight we make is that performing certain swaps keeps the shuffle close to uniform, which allows us to recover a uniform shuffle at the end of the procedure.

The other complication is that reads and writes on the input underlying the buffer must abide by the concurrency requirements the buffer demands. This requires a modification of how swaps are handled, but we will see that it is not hard to do.

\subsection{The Algorithm}

We first consider the parallel Knuth-shuffle algorithm given in \cite{highly-parallel}, which is Algorithm \ref{parallel-Knuth-shuffle} with parameters $r = 1, s = n$.

\begin{algorithm2e} [ht] 
\SetKwBlock{DoParallel}{do in parallel}{end}
\SetKwFor{ParFor}{parallel for}{do}{endfch}
\SetKwFor{ParForEach}{parallel foreach}{do}{endfch}
\caption{Parallel Knuth-shuffle$(A, r, s)$\Comment{Process swaps at indices $[r,s]$}}\label{parallel-Knuth-shuffle} 
$R \gets [-1,\dots,-1], S \gets [(-1,-1),\dots,(-1,-1)]$ \Comment{Initialize reservation and swap arrays} \\
\ParFor{$i \gets s$ to $r$}{
    $H[i] \xleftarrow{R} \{1,\dots,i\}$ \\
    $S[i] \gets (i,H[i])$
}
\While{$S$ is nonempty}{
    \ParForEach{swap $(s,H[s])$ in $S$}{
        $R[s] \gets \max(R[s],s)$ \\
        $R[H[s]] \gets \max(R[H[s]],s)$
    }
    \ParForEach{swap $(s,H[s])$ in $S$}{
        \If{$R[s] = s$ and $R[H[s]] = s$} {
            swap$(A[s],A[H[s]])$ and mark swap as finished
        }
    }
    Filter unfinished swaps from $S$ and reset $R$
}
\end{algorithm2e}

The main idea is that many of the swaps being made in the sequential Knuth-shuffle are independent of each other. As such, using reservations, we can enforce many swaps to happen simultaneously, and the resulting swaps that are made will be done in an order equivalent to that of the sequential Knuth-shuffle. The relaxed PIP algorithm of \cite{2021pip} utilizes the observation that for any $k$, one can instead carry out the algorithm in $n/k$ sequential rounds, where the $i$th round handles the swaps made by the elements from indices $n-ik$ to $n-(i+1)k$ (calls Parallel Knuth-shuffle$(A,n-(i+1)k,n-ik)$). To lower the memory footprint, reservations are handled using a hash-table of size $2k$. This algorithm uses $O(k)$ auxiliary memory. We will call this modification the $k$-Knuth-shuffle. The analysis in \cite{2021pip} holds for any $k$ to show that the expected work is $O(n)$ and the span is $O(n/k \log^2 n)$ \whp. Taking $k$ to be $O(n^{1 - \epsilon})$ gives the relaxed PIP algorithm.

One point that we remark on is that the parallel Knuth-shuffle algorithm satisfies stronger coherence properties then simply being decomposable. It turns out that the swaps can be executed in any order, and a uniformly random permutation will be generated regardless \cite{rand-perm-shared-memory}. The permutation theory we develop for the analysis of our algorithm also provides a proof of this fact (Corollary \ref{swap-corollary}). As a consequence, we may split up the swaps in any manner we desire.

Before we give the exact procedure, we briefly recall from Section \ref{restorable-buffers} the notion of buffers and the operations they support. A buffer with fixed auxiliary space of size $n/\log^2 n$ can be constructed on any contiguous range of the input of size $n/\log^2n + O(n/\log n)$. The auxiliary portion of this buffer supports reads and writes of $\log n$-bit words. Moreover, if we take the buffer to be adjustable, the underlying input of the buffer can be read and written (called simulated reads and writes). Performing these operations on the input underlying the auxiliary and encoding portions incurs $O(\log n)$ work and $O(1)$ work respectively. Moreover, auxiliary operations and encoding operations must be done disjointly.

We now describe the algorithm in detail. In the first stage, we set up a buffer with $O(n/\log^2 n)$ auxiliary, which can be done in an $O(n/\log n)$ suffix of the array. Using the buffer auxiliary space, we perform the $(n/\log^2 n)$-Knuth-shuffle to process the swaps from the prefix of the array excluding the buffer. This randomly shuffles the first $n-O(n/\log n)$ elements. We then restore the buffer.

In the second stage, we set up an adjustable buffer with $O(n/\log^2 n)$ auxiliary, this time in an $O(n/\log n)$ prefix of the array. We carry out the remaining swaps using a modified $n/\log^2 n$-Knuth-shuffle. The modification is necessary to accommodate the concurrency demands of the adjustable buffer. First of all, any swap involving the buffer auxiliary portion is made using simulated reads and writes and analogously for swaps involving the encoding portion. Second, the way reservations are made changes: if a swap targets an encoding pair, the swap must reserve both indices of that pair. Third, the execution of reserved swaps are now done in two steps. The reserved swaps which target indices in the auxiliary portion of the buffer are first performed in parallel, followed by the parallel execution of the other reserved swaps. After this modified $n/\log^2 n$-Knuth-shuffle finishes, we restore the buffer and, in parallel, uniformly randomly transpose each pair from the encoding portion of the restored buffer. This concludes the algorithm.

We now analyze the complexity of our random permutation algorithm. The first stage is $O(n)$ expected work and $O(\poly)$ span \whp. The buffer takes $O(n)$ work and $O(\log n)$ span to set up and restore. The remaining work and span are determined by the $n/\log^2 n$-Knuth-shuffle using the auxiliary provided by the buffer. As a buffer is used for auxiliary space, the work and span are identical to that of the normal $n/\log^2 n$-Knuth-shuffle, which is $O(n)$ expected work and $O(\poly)$ span \whp. For the second stage, the buffer set-up and restoration cost is identical to the first stage. The modification to the $n/\log^2 n$-Knuth-shuffle of processing the remaining elements does not change the complexity. However, simulated auxiliary operations may now be performed, but this can contribute at most an additional multiplicative $\log n$ factor (since simulated reads and writes are $O(\log n)$ work). As we are only processing $O(n/\log n)$ elements, we conclude that the modified $n/\log^2 n$-Knuth-shuffle contributes $O(n)$ expected work and $O(\poly)$ span \whp. Finally, the random transpositions is a $O(n/\log n)$ work and $O(\log n)$ span operation. Altogether, we conclude that the random permutation algorithm is $O(n)$ work in expectation and $O(\poly)$ depth \whp.

\subsection{Correctness}\label{correctness}

We now prove the correctness of the algorithm. The essential idea is to view random swaps on the input as random variables valued in the permutation group of the indices $\{1,\dots,n\}$. The input starts as the constant random variable whose only value is the identity permutation. Then executing some random swap on indices of the input is interpreted as multiplying by the corresponding random permutation. Under this perspective, the correctness of the algorithm is equivalent to the random variable given by the product of all swaps being uniformly distributed. The most important random swaps being made in every variant of Knuth shuffle is the random permutation given by the transposition $(I_i \ i)$, where $I_i$ is uniformly sampled from $\{1, \dots, i\}$. Due to its importance, we will call this random swap $S_i$. The correctness of the sequential Knuth-shuffle is then simply the assertion that
\[
    S_n S_{n-1} \cdots S_1 \sim U_n,
\] where $U_n$ denotes the uniformly random permutation on $\{1, \dots, n\}$. Similarly, the parallel Knuth-shuffle also produces a random permutation equivalent to $S_n S_{n-1} \cdots S_1$; this is because the way in which reservations are made guarantees that the order of the swaps in parallel is consistent with the sequential order.

In our algorithm, the buffer in the second stage adds additional complexity due to simulated write operations. However, each of these additional swaps necessitated by the simulated writes have the property that they are valued in permutations which only transpose disjoint encoding pairs. We say such a permutation is an \textit{encoder}. Let $E$ be $\prod_p E_p$, ranging over all encoding pairs $p$, where $E_p$ uniformly randomly transposes the indices of the pair $p$. We call $E$ the \textit{uniform encoder}.

Our algorithm produces a random permutation equivalent to the following product of random permutations
\[
    U_k Y' (S_{n} Y_n) (S_{n-1} Y_{n-1}) (S_{n-2} Y_{n-2}) \cdots (S_{k+1} Y_{k+1}) E.
\]
Here $U_k$ is the uniformly random permutation on indices $\{1,\dots,k\}$ built by the first stage of the algorithm. The random permutation $Y'$ corresponds to the buffer set-up in stage 2, and is an encoder. Each pair $S_j Y_j$ corresponds to the series of swaps performed in processing the swap $S_j$ using simulated writes. Each $Y_j$ is an encoder. We can write the permutation as a product in the displayed order again because of how reservations are made. Also note that $E$ is independent of all the other permutations appearing in the product and each $S_i$ is independent of all the permutations that precede it. We need to show that this product is equivalent to the uniform distribution. To start, we state four lemmas, all of which essentially follow from standard exposure arguments.

\begin{lem}\label{casework-rand-perm}
    The product $U_k E$ is identically distributed to $U_k$.
\end{lem}
\begin{proof}
    By exposing on $E$ and using the fact that $U_k$ is independent from $E$, we find
    \begin{align*}
        \Pr[U_k E = \sigma] &= \sum_{\tau} \Pr[U_k = \sigma \tau^{-1} \mid E = \tau] \Pr[E = \tau] \\
        &= \sum_{\tau} \Pr[U_k = \sigma \tau^{-1}] \Pr[E = \tau] \\
        &= \sum_{\tau} \frac{1}{k!}\Pr[E = \tau] = \frac{1}{k!} \sum_{\tau} \Pr[E = \tau] \\
        &= \frac{1}{k!} = \Pr[U_k = \sigma]. \qedhere
    \end{align*}
\end{proof}

\begin{lem}\label{ext-iid}
    Let $X, Y, Z$ be any random permutations. If $X$ is independent from both $Y$ and $Z$, and $Y \sim Z$, then $XY \sim XZ$ and $YX \sim ZX$.
\end{lem}
\begin{proof}
    Exposing on $X$, and using the independence of $X$, we deduce
    \begin{align*}
        \Pr[XY = \sigma] &= \sum_{\tau} \Pr[Y = \tau^{-1}\sigma \mid X = \tau] \Pr[X = \tau] \\
        &= \sum_{\tau} \Pr[Y = \tau^{-1}\sigma] \Pr[X = \tau] = \sum_{\tau} \Pr[Z = \tau^{-1}\sigma] \Pr[X = \tau] \\
        &= \sum_{\tau} \Pr[Z = \tau^{-1}\sigma \mid X = \tau] \Pr[X = \tau] = \Pr[XZ = \tau].
    \end{align*}
    This concludes that $XY \sim XZ$, and a symmetric argument yields $YX \sim ZX$.
\end{proof}

\begin{lem}\label{xyz-lemma}
    Let $X$ be any random permutation and $Y$ an encoder (which potentially depends on $X$), both independent from $E$. Then $X$ and $YE$ are independent. Moreover, $YE \sim E$, and thus $XYE \sim XE$.
\end{lem}
\begin{proof}
    Let $H$ denote the subgroup of permutations which only transpose encoding pairs, so $E$ is uniformly distributed over $H$. Also $Y$ takes values in $H$, as $Y$ is an encoder, so $YE$ is also supported in $H$. It thus suffices to consider $\tau \in H$. Then by exposing on $Y$, we see that
    \begin{equation*}
        \begin{split}
        \Pr[YE = \tau \mid X = \sigma] &= \sum_{\gamma \in H}  \Pr[E = \gamma^{-1} \tau \mid X = \sigma, Y = \gamma] \Pr[Y = \gamma \mid X = \sigma] \\
        &= \sum_{\gamma \in H} \Pr[E = \gamma^{-1}\tau] \Pr[Y = \gamma \mid X = \sigma] = \frac{1}{|H|} = \Pr[E = \tau].
        \end{split}
    \end{equation*}
    and
    \begin{equation*}
        \begin{split}
        \Pr[YE = \tau] &= \sum_{\gamma \in H}  \Pr[E = \gamma^{-1} \tau \mid Y = \gamma] Y = \gamma] \Pr[Y = \gamma] \\ &= \sum_{\gamma \in H}  \Pr[E = \gamma^{-1} \tau \mid Y = \gamma] \Pr[Y = \gamma] = \frac{1}{|H|} = \Pr[E = \tau].
        \end{split}
    \end{equation*}
    We conclude the lemma.
\end{proof}

\begin{lem}\label{commuter}
    If $X$ is a random permutation which only acts nontrivially on indices $\{1,\dots,i-1\}$ and is independent of $S_i$, then $S_iX \sim X S_i$.
\end{lem}
\begin{proof}
    Note for any permutation $\sigma$ and transposition $(j \ i)$, we have $\sigma (j \ i) = (\sigma j \ \sigma i) \sigma$. But as $X$ is supported on permutations which fix $i$, and $I_i$ is uniformly random, exposing on $X$ and $S_i$ gives
    \begin{align*}
        \Pr[XS_i = \tau] &= \sum_{j \leq i} \Pr[X (j \ i) = \tau \mid S_i = (j \ i)] \Pr[S_i = (j \ i)] \\
        &= \frac{1}{i}\sum_{j \leq i} \Pr[X (j \ i) = \tau] = \frac{1}{i}\sum_{j \leq i} \sum_{\sigma \in \supp X} \Pr[\sigma (j \ i) = \tau] \Pr[X=\sigma] \\
        &= \frac{1}{i} \sum_{j \leq i} \sum_{\sigma \in \supp X} \Pr[(\sigma j \ i) \sigma = \tau]\Pr[X=\sigma] \\ &= \frac{1}{i} \sum_{\sigma \in \supp X} \sum_{j \leq i} \Pr[(\sigma j \ i) \sigma = \tau]\Pr[X=\sigma] \\
        &= \frac{1}{i}\sum_{j \leq i} \sum_{\sigma \in \supp X} \Pr[(j \ i) \sigma = \tau]\Pr[X=\sigma] = \Pr[(j \ i) X = \tau] \\
        &=\sum_{j \leq i} \Pr[(j \ i) X  = \tau \mid S_i = (j \ i)] \Pr[S_i = (j \ i)] = \Pr[S_i X = \tau]. \qedhere
    \end{align*}
\end{proof}

\begin{cor}\label{swap-corollary}
    For any $i, j$, we have $S_iS_j = S_jS_i$.
\end{cor}

We now prove the desired result inductively. Let $X^{l}$ denote the product that $S_nY_n \cdots S_{l}Y_{l}$ and let $S^l$ denote the product $S_l \cdots S_{k+1}$ (we will consider $S^k$ to be trivial). We show that 
\[
    X^{l+1} E S^l \sim X^{l+2} E S^{l+1}.
\]
Consider that $X^{l+1} E S^l = X^{l+2} S_{l+1} Y_{l+1} E S^l$. Note that $S^{l}$ is independent from every other term in this product. By Lemma \ref{xyz-lemma}, we have $X^{l+2} S_{l+1} Y_{l+1} E \sim X^{l+2} S_{l+1} E$. Moreover, by Lemma \ref{commuter}, we have $S_{l+1} E \sim E S_{l+1}$, and both $S_{l+1} E$ and $E S_{l+1}$ are independent from $X^{l+2}$, so Lemma \ref{ext-iid} applies to get $X^{l+2} S_{l+1} E \sim X^{l+2} E S_{l+1}$. Finally, as $S^l$ is independent from everything, we apply Lemma \ref{ext-iid} to conclude that \[X^{l+1} E S^l \sim X^{l+2} S_{l+1} E S^l \sim X^{l+2} E S_{l+1} S^l = X^{l+2} E S^{l+1}.\]

As our original permutation was $U_k Y' X^{k+1} E S^k$, we conclude that it is equivalent to $U_k Y' E S^{n}$. As $Y'$ is an encoder, we get that $Y' E \sim E$, and also $U_k Y' E S^{n} \sim U_k E S^{n}$. Now $U_k$ is uniformly distributed on the permutations acting on the first $k$ indices, which is independent from $E$, which acts on the first $k$ indices. As $U_k E \sim U_k$ by Lemma \ref{casework-rand-perm}, we see that $U_k Y' E S^{n} \sim U_k S^n$. But we can write $U_k$ as $S_k S_{k-1} \dots S_1$, and then $U_k S^n$ is the product $S_k \dots S_1 S_n \cdots S_{k+1}$. Now, Corollary \ref{swap-corollary} says this product is equivalent to $S_n \cdots S_1$, which is exactly the same as the random permutation produced by the sequential Knuth-shuffle. We thus conclude the correctness of the algorithm.
\section{An Application to Integer Sorting}\label{add-apps}

In this section, we demonstrate how our techniques can apply to other problems, namely integer sorting. The problem of sorting fixed-length integers is fundamental in computer science and has been studied extensively. We outline a procedure to solve this problem, where the max integer $r$ is $O(n^c)$, in-place using the techniques we have developed.

We first partition the array in parallel in-place across the median using the algorithm given in \cite{kuszmaul2020inplaceparallelpartitionalgorithmsusing}. Note this algorithm can be used initially to find the median with an in-place parallel quick select. Now we do the same for both halves in parallel, and so on, until each slice is of size $O(n/\log n)$. This requires $O(\log \log n)$ rounds where each round is linear work, so this step takes $O(n \log \log n)$ work. This reduces the problem to $O(\log n)$ disjoint subproblems, each of size $O(n/\log n)$. We now describe how to sort each subproblem by setting up buffers in the remaining subproblem to facilitate the use of a non-in-place MSD sorting algorithm.

To start, we process the first subproblem. In the space occupied by every other subproblem, we set up maximally sized restorable buffers. This gives $O(n/\log n)$ auxiliary space. We partition across medians again in the first problem until the subsubproblems have size less than the auxiliary (only a linear amount of work needs to be done to achieve this). Next, we use the non-in-place MSD radix sort algorithm given by \cite{10.1145/3627535.3638483} to sort this subproblem. With this subproblem solved, we can now set up a restorable buffer within it. We then restore the buffer in the second subproblem and repeat the procedure we used to solve the first. This continues until we have solved all subproblems. The span contributed by this step is $O(\poly rn)$ and the work is $O(n\sqrt{\log r})$. In total the work and span for this algorithm is $O(n\sqrt{\log n} + n\log\log n)$ and $O(\poly n)$ respectively.


\section{Implementation Optimizations} We implement our parallel in-place merging algorithm from Section \ref{merge section}. For simplicity of implementation, we assume that our input is divisible by a given block size. We now discuss several optimizations that improve practical performance in our in-place merge algorithm.
\subsection{Enlarging Blocks} In Section \ref{tphasemerge}, we chose a block size of $O(\log n)$ to fit the strong PIP model. However, for large inputs, the parallel loops used in the EndMerge algorithm (Algorithm \ref{EndMerge}) incur nontrivial scheduling overhead, since the number of blocks, and hence the number of iterations, is $n / \log n$. Moreover, encoding and decoding the values stored in each block becomes costly when there are a large number of blocks. To improve performance, we enlarge the block size. Empirically, in inputs exceeding 100 million elements, we found that setting the block size to 4000 yielded the best performance. While increasing the block size does raise the per-thread stack usage, we observed that the total space overhead remained negligible. For example, on a 96-core machine, using a block size of 4000 required only ${\sim}158$ KB of additional stack space (total across all threads), which we believe is reasonable in practice.
\subsection{Precomputing Randomness} Each iteration of EndMerge computes a random bit for each active block. Since the random bits used by a single block across different iterations are independent, we can precompute them. For example with 32-bit integers, in the initialization loop of EndMerge, we simply compute and encode a random 32-bit integer for each block. Then, in iteration $i$, each block simply retrieves its random bit from the $i$-th bit in that precomputed value.

The main benefit of this approach is that it removes one parallel loop per iteration of the while-loop. Instead of generating random bits during each iteration, all randomness is generated once upfront. Consequently, the scheduling overhead of those extra parallel loops is avoided, resulting in a measurable speedup on large inputs.
\subsection{Coarsening EndMerge} To reduce the redundant work done by non active blocks in EndMerge, we perform a fixed number of iterations of the while loop and complete the remaining swaps using a naive symmetry breaking. This has the additional benefit of removing the need for the Done function. 

To complete the remaining swaps, we do parallel loop over all blocks, and for each active block $X[i]$, we traverse along end-sorted-position pointers until returning back to index $i$. If during this process $X[i]$ traverses to a block $j$ such that $j < i$, then $X[i]$ marks itself. In each remaining cycle, one block will not mark itself. This block becomes the `leader' and processes its cycle sequentially. We mark blocks by encoding an additional bit per block, or reusing a previously encoded bit. 

We empirically determined that setting the number of iterations to $O(\log N)$ ensures that each iteration removes a large fraction of the remaining active blocks so that the overhead of completing the remaining swaps is negligible.

\subsection{Using a Buffer} Recall that in each iteration, every active block $X[i]$ reads the coin flip of its target whose index is given by $T_b[i]$. Reading $T_b[i]$ is an $O(\log n)$ work operation, and while it doesn't affect the overall theoretical complexity, a non-trivial number of blocks will perform this operation repeatedly throughout EndMerge. To avoid this, each block uses a restorable buffer to store $T_b[i]$. Concretely, the first $O(\log n)$-bits of the first element in each block is encoded using the next $O(\log n)$ pairs and then overwritten with $T_b[i]$. Then to read $T_b[i]$ in each iteration, $X[i]$ simply performs an array access is $O(1)$ time.

\subsection{On Removing Divisibility} As previously mentioned, our implementation assumes that input size is divisible by a given block size (e.g. 4000). Here we provide a sketch of a simple and practical strategy for removing this assumption.

The procedure described in Section \ref{removing divisibility} is nice in theory, but is unlikely to be practical. Instead, we allow ourselves to use the heap to a small degree. Let $b$ be a given block size. Letting $A$ and $B$ be our arrays, we consider the largest suffix of $A$ divisible by $b$, and the largest prefix of $B$ divisible by $b$. We then use the heap to extend the prefix of $A$ by elements $-\infty$ and the suffix of $B$ by $\infty$ to simulate a $b$-sized block. Thus we may now imagine that $A$ and $B$ are divisible by $b$. We run the same algorithm as before but with the modification that we keep track of which blocks currently contain the `imaginary' elements. As there can only ever be one such block (the $\infty$ elements will never move), doing this is easy. Further, this approach uses only $O(2b)$ heap-allocated memory which for reasonable $b$ is completely negligible.

\section{Experiments} \label{exp section} We implemented our parallel in-place merging algorithm in C++ and tested its performance. We use ParlayLib \cite{10.1145/3350755.3400254} to support fork-join parallelism.

\subsection{Experimental Setup} We run our experiments on a 96-core machine (192 hyper-threads) with $6 \times 2.10$ GHz Intel Xeon Platinum 8160 CPUs (each with $16.5$MB L3 cache) and $1.5$TB of main memory. We compile using the g++ compiler (version 11.4.0) with the -O2 flag. We compare our algorithm to the non-in-place merging implementation provided in the Problem Based Benchmark Suite (PBBS) \cite{10.1145/2312005.2312018}, which is a collection of highly optimized parallel algorithms and implementations.
\subsubsection{Input representation.} We test both implementations on $32$-bit fixed-width integers (i.e. the uint32\_t type in C++). We test input sequences that are randomly generated. 
\subsubsection{Block size.} We use blocks of size 4000 which, for inputs less than 2 billion, provides a good trade off between the total number of blocks and the size of each block. 

\subsection{Performance Analysis} Figure \ref{run time comparison} contains the running times of the in-place parallel merging (IP Merging) and PBBS merging implementations. Figure \ref{fig:relative-performance} displays the corresponding slow down of IP merge. On small input sizes (< 50M) the overhead of encoding techniques and generally the complexity of IP merge results in poor performance comparatively. However, for larger inputs, IP merge becomes competitive. This is likely due to its small memory overhead.

\subsubsection{Space usage.} We measure space usage using Valgrind's Massif tool. Table \ref{tab:ip-merge-mem} shows the total memory usage of IP merging on an input of 2 billon integers. We see that the auxiliary space usage of IP merging is negligible. We also note that the majority of the memory usage comes from the ParlayLib work-stealing scheduler (${\sim}10$ MB).
\begin{table}[ht]
\centering
\begin{tabular}{lr}
\hline
\textbf{Input Memory Size} & 1,073,741,824  \\
\textbf{Peak Total Memory} & 1,086,896,064  \\
\textbf{Overhead (\%)}     & 1.22\%                         \\
\textbf{Peak Stack Usage}  & 161,816        \\
\hline
\end{tabular}
\caption{Memory usage (in bytes) for IP Merging on input of size 2 billion.}
\label{tab:ip-merge-mem}
\end{table}

\subsubsection{Scalability.} IP merging has good scalability to both the number of threads and input sizes as seen in Figure \ref{run time comparison} and \ref{fig:scalability-comparison}.  

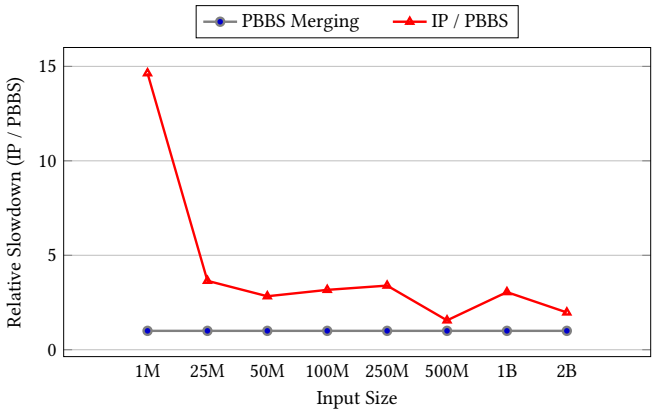
\begin{figure}[ht]
\scalebox{0.75}{%
  \begin{minipage}{\textwidth}
\begin{tikzpicture}
\begin{axis}[
    major x tick style = transparent,
    width=12cm,
    height=7cm,
    xlabel={Input Size},
    ylabel={Relative Slowdown (IP / PBBS)},
    ymajorgrids=true,
    enlarge x limits=0.2,
    symbolic x coords={1M,25M,50M,100M,250M,500M,1B,2B},
    xtick=data,
    legend cell align={left},
    legend style={
        at={(0.5,1.03)},
        anchor=south,
        legend columns=2,
        /tikz/every even column/.append style={column sep=0.5cm}
    },
]

\addplot+[mark=*, very thick, color=gray] coordinates {
    (1M, 1)
    (25M, 1)
    (50M, 1)
    (100M, 1)
    (250M, 1)
    (500M, 1)
    (1B, 1)
    (2B, 1)
};

\addplot+[mark=triangle, very thick, color=red] coordinates {
    (1M, 14.635)
    (25M, 3.649)
    (50M, 2.831)
    (100M, 3.169)
    (250M, 3.395)
    (500M, 1.554)
    (1B, 3.051)
    (2B, 1.974)
};

\legend{PBBS Merging, IP / PBBS}

\end{axis}
\end{tikzpicture}
\caption{Relative performance (IP time / PBBS time) vs. input size. PBBS is baseline at y = 1.}
\label{fig:relative-performance}
\end{minipage}%
}
\end{figure}

\begin{figure}[ht]
\scalebox{0.75}{%
  \begin{minipage}{\textwidth}
\begin{tikzpicture}
\begin{axis}[
    major x tick style = transparent,
    width=12cm,
    height=7cm,
    xlabel={Threads},
    scaled y ticks = false,
    ylabel={Time (\(\mu s\))},
    ymajorgrids=true,
    enlarge x limits=0.2,
    symbolic x coords={1,4,16,64,96,96HT},
    xtick=data,
    legend cell align={left},
    legend style={
        at={(0.5,1.03)},
        anchor=south,
        legend columns=2,
        /tikz/every even column/.append style={column sep=0.5cm}
    },
]

\addplot+[mark=*, very thick, color=gray] coordinates {
    (1,   3040850)
    (4,   873144)
    (16,  210963)
    (64,   72052)
    (96,   45981)
    (96HT,  38324)
};
\addlegendentry{PBBS Merging}

\addplot+[mark=triangle, very thick, color=red] coordinates {
    (1,   4539910)
    (4,   1196750)
    (16,  391647)
    (64,  202730)
    (96,  126598)
    (96HT, 106874)
};
\addlegendentry{IP Merging}

\end{axis}
\end{tikzpicture}
\caption{Running times for PBBS Merging and IP Merging (500M elements) across varying thread counts.}
\label{fig:scalability-comparison}
\end{minipage}
}
\end{figure}
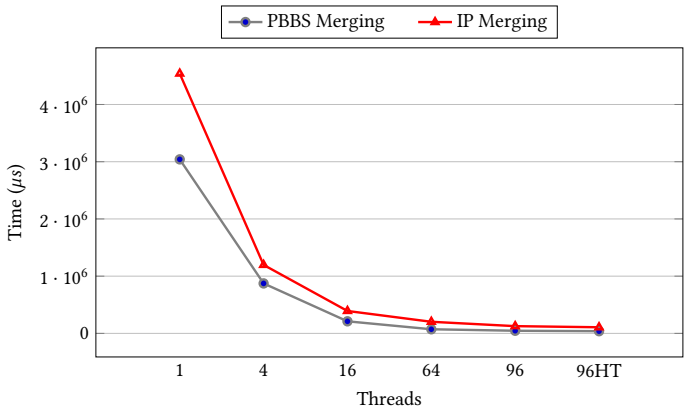




\section{Conclusion.} In this paper, we provided work-efficient strong PIP algorithms for the problems of merging two sorted sequences and randomly shuffling a sequence. Moreover, we propose a general encoding technique in restorable buffers that we believe has applications beyond what we've done in this paper. Ultimately, in this paper we make assumptions that we have to ability to modify the input and further that the input has an initial structure that can be exploited to gain access to more bits (i,e, bit-stealing). While the former assumption is quite reasonable, there are a number of simple problems that do no satisfy the latter assumption. Take for example the problems of list ranking and stable partition (note that the algorithm given in \cite{kuszmaul2020inplaceparallelpartitionalgorithmsusing} is unstable). Fundamentally, the techniques we used in this paper seemingly do not apply to these problems and thus devising work-efficient strong PIP algorithms for list ranking and stable partition is an interesting line of future work.

\begin{acks}
We thank Laxman Dhulipala for letting us use his group's machine for our experiments and for insightful discussions which initiated this paper.
\end{acks}

\bibliography{refs.bib}{}
\bibliographystyle{plain}

\end{document}